\documentclass[12pt]{iopart}
\usepackage{iopams, amsthm}
\usepackage{cite}
\usepackage{url}

\theoremstyle{plain}
\newtheorem{thm}{Theorem}
\newtheorem{prop}[thm]{Proposition}
\newtheorem{lem}[thm]{Lemma}
\newtheorem{cor}[thm]{Corallary}
\theoremstyle{definition}
\newtheorem{defn}[thm]{Definition}

\theoremstyle{remark}
\newtheorem{rmk}[thm]{Remark}

\begin{document}
\title[Regular Hyperbolicity, DEC, and Causality]{Regular hyperbolicity, dominant energy condition and causality for
Lagrangian theories of maps}
\author{Willie Wai-Yeung Wong}
\address{Department of Pure Mathematics and Mathematical Statistics\\
University of Cambridge\\ Cambridge, UK}
\ead{\mailto{ww278@dpmms.cam.ac.uk}}

\begin{abstract}
The goal of the present paper is three-fold. First is to clarify the
connection between the dominant energy condition and hyperbolicity
properties of Lagrangian field theories. Second is to provide further
analysis on the breakdown of hyperbolicity for the Skyrme model,
sharpening the results of Crutchfield and Bell and comparing against
a result of Gibbons, and provide a local well-posedness result for the
dynamical problem in the Skyrme model. Third is to provide a short 
summary of the
framework of regular hyperbolicity of Christodoulou for the relativity
community. In the process, a general theorem about dominant energy
conditions for Lagrangian theories of maps is proved, as well as
several results concerning hyperbolicity of those maps. 
\end{abstract}

\pacs{1110Ef, 1239Dc}


\section{Introduction}
In the study of classical field theories, an oft-imposed ``physical
assumption'' is the \emph{dominant energy condition}, which requires
that the Einstein-Hilbert stress-energy tensor $T$ appearing
on the ``wood'' side \cite{Einste1936} 
of Einstein's equation to have the property that,
for any future causal vectors $X$ and $Y$, the contraction
$T(X,Y)$ is non-negative. \emph{A posteriori} this condition
seems reasonable, in view of the results that can be derived from it.
Two of the most prominent examples are, of course, the Singularity
Theorem of Penrose \cite{Penros1965} and the Positive Mass Theorem
\cite{SchYau1981, Witten1981}. However, these results are purely
results in Lorentzian geometry: that is, Einstein's equation is wholly
unnecessary, except to transfer the dominant energy condition from the
``wood'' side to the ``marble'' side of the equation. In other words,
those theorems could equally well have been stated without reference
to general relativity, but merely with certain positivity conditions
about the Einstein tensor $G = Ric - \frac12 Rg$. 

The typical \emph{a priori} justification
for the imposition of the dominant energy condition is some
heuristic requiring that the flow of energy be at a speed less than
that of gravity\footnote{If one is willing to assume a linear or 
semilinear theory of electromagnetism, then also speed of light.}, that
is, a property on the classical field variously known as 
\emph{finite speed of propagation}, \emph{domain of dependence} or
\emph{causality} (see, e.g.\ \cite{Wald1984}). Perhaps the most
well-known in this regard is a theorem of Hawking \cite{HawEll1973}
which makes precise this notion: that if a matter field
satisfies the dominant energy condition and the energy is strongly
coercive (that is, $T(X,Y) = 0$ for non-zero vectors $X$, $Y$ only when the
matter field vanishes), then if the matter field vanishes in some
space-time region, it must also vanish in its domain of dependence. 
Typically, a proof of the domain of dependence property uses an
\emph{energy estimate}, which is the curved space-time version of the
law of energy conservation on Minkowski space-time. And therefore,
it is usually expected that a domain of dependence property comes hand
in hand with the hyperbolicity of the associated matter evolution. 

This is, however, not always the case. The principal obstruction is
that Hawking's theorem only guarantees the finite speed of propagation
of ``perturbations of vacuum''. That is, it essentially only
guarantees that the edge of vacuum cannot recede faster than the speed
of gravity. For semilinear field theories where the kinematics in the
high frequency limit is always governed by the space-time metric,
there can be no distinction between perturbations of vacuum and
perturbations of a given solution. And hence the argument used by
Hawking to establish his domain of dependence theorem generally
indicates hyperbolicity for the associated matter field. For 
quasilinear field theories, however, the
strong self-interaction means that the kinematics close to a vacuum
background can be significantly different from that around a ``large''
solution. And for these types of theories, dominant energy condition
is not sufficient to guarantee hyperbolicity. 

A prime example of this difficulty is illustrated by two results
related to the Skyrme model of nuclear physics. Motivated by unstable
numerical simulations, Crutchfield and Bell demonstrated
\cite{CruBel1994} that highly boosted background solutions to the
Skyrme model is linearly unstable under perturbations (that there
exists an exponentially growing mode). On the other hand, it was shown
by Gibbons \cite{Gibbon2003} that the Skyrme model in fact enjoys the
dominant energy condition. In the current paper, a more detailed
analysis of the breakdown of hyperbolicity for the Skyrme model will
be presented. The analysis is based on the Christodoulou's regular
hyperbolicity framework \cite{Christ2000}.

(One may also ask whether the reverse implication is true: that 
hyperbolic matter models will always enjoy a dominant energy condition. 
A simple example showing its falsehood is the linear wave equation
with a negative potential $\Box u = -u_{tt} + \triangle u = Vu$, where
the potential $V \leq 0$. If one were to desire theories without
external potentials, one can also consider the focusing nonlinear
wave equation $\Box u = - |u|^p u$, the hyperbolicity, or local 
well-posedness, of which
is well-known (see, e.g.\ \cite{ShaStr1998}). In regions
where the $|u|$ is large while approximately constant, the associated
Einstein-Hilbert stress-energy tensor for either of the above examples
violates the dominant energy condition.)

In the present paper the connections between the dominant energy
condition and hyperbolicity will be studied in the context of
Lagrangian theory of maps. An emphasis will be placed on the
insufficiency for mutual implication. The reasons for the focus on such
matter models are twofold. First is its general applicability. Many
models of mathematical physics can be cast in the framework of
Lagrangian theory of maps. Starting from the simple linear wave
equation, which can be regarded as a map from Minkowski space to the
real line or to the complex plane, we can modify the target space to a
general Riemannian manifold and obtain what is called the nonlinear
$\sigma$-model in the physics literature, or the wave-map system in
mathematics. This system is itself interesting as models in
high-energy physics (see \cite{Jost2009} and references therein) 
or as symmetry reductions from Einstein's
equations in general relativity (see, e.g.\ \cite{Camero1991,
ChoMon2001, Mazur1982}). As a semilinear modification to the standard
wave equation with a geometric interpretation, the well-posedness
properties (both global and local, and in both small and large data
regimes) have been well-studied, see
\cite{RodSte2010, RapRod-p2009, Tao-pHW3, Tao-pHW4, Tao-pHW5,
Tao-pHW6, Tao-pHW7, KriSch-p2009} for some recent progress on global
regularity and singularity formation in the large data regime, and
\cite{Choque1987, KlaMac1995, KlaMac1997, KlaSel1997, Sideri1989,
Tao2001, CaShTa1998, Shatah1988} for a sample of classical results in
this area. 

In this paper we will focus on \emph{quasilinear} modifications to the
wave-map system, whose dynamics is comparatively less well-studied. Such
generalisations also have wide physical applications, with examples
in the nonlinear $\sigma$-model hierarchy including the Skyrme model
\cite{Skyrme1954, Skyrme1961, ManSut2004}, the membrane equation 
\cite{Hoppe1994, Lindbl2004}, a Born-Infeld type model in cosmology
\cite{Johnso2003, Leigh1989}, and models of hydro- and elastodynamics
\cite{BeiSch2003, Bressa1978, CarQui1972, Christ1998, KijMag1992,
Salt1971, Slobod-p2010, Tahvil1998}. The models of dynamics in a
continuous medium are particularly interesting in this context, as
generally a physical assumption in such models is that the particle
world-lines are time-like, a condition necessary to guarantee the
causality of the matter model. It will
be shown in this paper that one can construct examples of equations of
states for which the dominant energy condition is satisfied
\emph{independently} of whether the physical constraint is imposed.
This further reinforces the idea that the domain of dependence theorem
of Hawking is only a statement about vacuum perturbations. 

The second reason for considering Lagrangian theories of maps is
purely technical. For Lagrangian field theories, the Einstein-Hilbert
stress-energy can be defined via a variational procedure on the
Lagrangian density. This allows for general and efficient calculations to
check the dominant energy condition. Furthermore, in the context of
Lagrangian theories of maps, the regular hyperbolicity framework of 
Christodoulou \cite{Christ2000} provides a powerful while
algebraically simple characterisation of local well-posedness.
Therefore we will consider only such matter models for ease of
discussion. 

The paper is organised as follows. In Section \ref{sect:lag}
we review the Lagrangian field theory of maps and give some examples
that have appeared in the literature. In Section
\ref{sect:dec} a geometric method is described for computing the
Einstein-Hilbert stress-energy tensor for a large class of maps which
includes the physically interesting models described above. The method
provides an easy way to verify the dominant energy condition for these
maps; we recover the result of Gibbons \cite{Gibbon2003} as a special
case. In Section \ref{sect:rh} we briefly describe the philosophy and 
method of regular hyperbolicity of Christodoulou
\cite{Christ2000}, and recall the notion of canonical stress. 
Here only the basic ideas behind the theory of regular hyperbolicity
will be sketched, the focus being on its application. In
a forthcoming paper with Jared Speck \cite{SpeWon-pPP}, a detailed 
gentle introduction to regular hyperbolicity will be given, along 
with some simple extensions that were alluded to, without proof, in 
Christodoulou's monograph. And in Section \ref{sect:ex}
we apply the theory to the problem of hyperbolicity of
the Skyrme model. 

\section{Lagrangian theory of maps}\label{sect:lag}
Throughout we let $(M,g)$ be an $m+1$ dimensional Lorentzian manifold,
where sign convention is taken to be $(-,+,+,\ldots)$; and we let
$(N,h)$ be an $n$ dimensional Riemannian manifold. $M$ represents the
physical space-time (often taken to be 3 space and 1 time dimensions,
though we make no such restrictions here), while $N$ represents the
internal structure of the field. In applications, for nonlinear
$\sigma$-models, $N$ is usually taken to be a Lie group or a
symmetric space; in dynamics of a continuous medium, $N$ is the
material manifold\footnote{In the fluids literature, the
material manifold is sometimes only required to be equipped with a
volume form; but as every paracompact smooth manifold admits a
Riemannian metric, and a Riemannian metric can realize any volume form
through conformal rescaling, it is of no loss in generality to assume
the material manifold is Riemannian.} with $n = m$.

Denote by $\phi:M\to N$ a continuously differentiable map. In field
theories $\phi$ gives the state of the field at a point in space-time,
whereas for continuum mechanics $\phi$ represents the coordinate
transformation between the Eulerian and Lagrangian pictures.
Then the action of $\phi$ can be used to pull back the
metric $h$ onto $M$ as a positive semi-definite quadratic form on
$TM$, which we write as
\[ \phi^*h(X,Y) = h(d\phi\cdot X, d\phi\cdot Y) \]
where the left hand side is evaluated at a point $p\in M$ and the
right hand side at the point $\phi(p)\in N$ for $X,Y\in T_p M$.

We define the (1,1)-tensor field $D^\phi$ by composing with the
inverse metric $g^{-1}$:
\begin{equation}\label{eq:def:strain}
D^\phi = g^{-1}\circ \phi^*h~.
\end{equation}
We will follow Manton and Sutcliffe \cite{ManSut2004} and call
this the \emph{strain tensor} for the map, the nomenclature taken from
the study of dynamics in a continuous medium, where, roughly speaking,
the trace of $D^\phi$ on a space-like hypersurface describes the local
deformation of the material. See \cite{KijMag1992, BeiSch2003, 
Tahvil1998} for more detailed discussions (note that our
definition here agrees with that of Tahvildar-Zadeh \cite{Tahvil1998}
if we impose the physical assumption that the map $\phi$ admits a
space-like simultaneous space, but our definition differs from that of
Kijowski and Magli \cite{KijMag1992}, who in addition to the above
physical assumption, also breaks the degeneracy by adding in the 
square of the particle velocity to make the strain tensor positive
definite). 

At a fixed point $p\in M$, the tensor $D^\phi$ defines a linear
transformation of the tangent space $T_pM$. We can consider its
eigenvalues. In the case that $g$ is a Riemannian metric, $D^\phi$ is
a self-adjoint operator on $T_pM$ relative to the (positive definite)
inner-product induced by the metric, and hence all the eigenvalues
would be real. For the Lorentzian case, the eigenvalues are in general
complex. Denote by $\{\lambda_1,\ldots,\lambda_k\}$ the non-zero
eigenvalues, counted with multiplicity. One easily sees that
\begin{equation}\label{eq:kupperbound}
 k \leq \mathrm{rank}(d\phi) \leq \min(m+1,n)~. 
\end{equation}
Recall the elementary symmetric polynomials
$\sigma_j(\{\lambda_1,\ldots,\lambda_k\})$ given by
\begin{equation}
\sigma_j(\{\lambda_1,\ldots,\lambda_k\}) = \sum_{1\leq \alpha_1 < \alpha_2
< \cdots < \alpha_j \leq k} \prod_{i = 1}^{j}\lambda_{\alpha_i}
\end{equation}
with $\sigma_0 = 1$ and $\sigma_j = 0$ for all $j > k$. The $\sigma_j$ correspond to
the coefficients of the characteristic polynomial for $D^\phi$, with
$\sigma_1 = \tr (D^\phi)$ and $\sigma_{m+1} = \det (D^\phi)$, and are algebraic
invariants of the tensor field $D^\phi$. Furthermore since $D^\phi$ is
a real linear transformation, the values of the $\sigma_j(D^\phi)$ are all
real. By an abuse of notation, we will write $\sigma_j(D^\phi)$ when we
mean the symmetric polynomials on the eigenvalues of $D^\phi$.

In this paper we will only consider Lagrangian field theories for maps
$\phi: (M,g)\to (N,h)$ where the action integral is given by
\begin{equation}\label{eq:genlag}
S = \int_M L(s(\phi), \sigma_1(D^\phi), \sigma_2(D^\phi), \ldots,
\sigma_{m+1}(D^\phi)) dvol_g~, \end{equation}
where $s:N\to\mathbb{R}$ is a non-negative scalar function on the
internal space $N$. 
In particular we require the only dependence of the Lagrangian on the
field $\phi$ be through the algebraic invariants $\sigma_j(D^\phi)$ and the
value of $\phi$ itself. (Note that by \eref{eq:kupperbound} any 
$\sigma_j$ not listed above as an argument
for $L$ is not dynamical.) Physically the dependence on $\sigma_j(D^\phi)$
and not other components of $D^\phi$ corresponds to
the assumption that the laws of physics are locally invariant under a
Lorentzian rotation in $O(1,m)$ of the tangent space $T_pM$ that fixes
the kernel of $D^\phi$. In relativistic elasticity where a space-like
simultaneous space is imposed, the relevant subgroup of $O(1,m)$ to
consider is the orthogonal group $O(n)$, and this condition on the
action corresponds to the
assumption that the material is \emph{homogeneous, isotropic, and
perfectly elastic} \cite{Tahvil1998}. The function $s$ corresponds to
the \emph{entropy per particle} in relativistic elasticity, and plays
the role of the symmetry-breaking \emph{mass term} in the Skyrme model. 

Observe that since the action itself only depends on $D^\phi$, which
is defined through only the first derivative of the map $\phi$, the
equations of motion given by applying the Euler-Lagrange equations to
\eref{eq:genlag} will be quasilinear second-order partial differential
equations. 

Some explicit examples of Lagrangian field theories from the
literature that fall in this class include:
\begin{description}
\item[Wave maps] can be described by the Lagrangian function
\begin{equation}
L = \sigma_1(D^\phi) = |\nabla\phi|^2
\end{equation}
where the norm is taken relative to both $g$ and $h$. 
\item[Relativistic elasticity] typically makes the additional
assumption that $d\phi$ is onto, and its kernel time-like (the
space-like simultaneous space assumption). Then interpreting $s$ as the
entropy per particle, the action \eref{eq:genlag} is the general form
for studying a homogeneous, isotropic, and perfectly elastic solid.
The case where $L = L(s,\sigma_n)$ is the special case representing
\emph{relativistic fluids}, $\sigma_n = |\phi^*dvol_h|^2$ being equal
to the squared $g$-norm of the pullback of the volume form on $N$. 
\item[Skyrme model] classically assumes $(M,g)$ to be the
4-dimensional Minkowski space with $(N,h)$ the standard 3-sphere,
though we can also consider the model without such constraints. The
Lagrangian function is (up to rescaling) \cite{Skyrme1954, Skyrme1961}
(see also \cite{ManRub1986, Loss1987})
\begin{equation}
L = \sigma_1(D^\phi) + \sigma_2(D^\phi) + s(\phi)~, 
\end{equation}
where $\sigma_2(D^\phi) = |d\phi\wedge d\phi|^2$, and $s(\phi) = m
\mathop{dist}_h(\phi,\phi_0)$ is a multiple of the geodesic distance
of $\phi$ from a fixed point $\phi_0\in N$, the constant $m$ being the mass
parameter. In general, Lagrangians of the form $L = \sigma_1 + \ldots$
are the nonlinear $\sigma$-models. 
\item[Born-Infeld type models] fix a large constant $b >0$.
Restricting the allowed maps $\phi$ to those whose eigenvalues of
$D^\phi$ have real parts at least $-b$, the action is given by
\begin{equation}
L = \sqrt{\det (b\cdot Id + D^\phi)} - \sqrt{\det(b\cdot Id)}~.
\end{equation}
This is the $\sigma$-model analogue of the Maxwell-Born-Infeld model
of nonlinear electrodynamics, see also \cite{Gibbon2003}.
\item[Membranes] as considered in \cite{Hoppe1994, Lindbl2004} can be
described by setting $(M,g)$ to Minkowski space of some fixed
dimension, and $(N,h)$ to the real line with canonical metric, and
taking 
\begin{equation}
L = \sqrt{ 1 + \sigma_1(D^\phi)}
\end{equation} 
while restricting consideration to those $\phi$ for which $L$ is
well-defined. It can also be viewed as an analogue to the minimal
surface equation for embedding into Minkowski space. 
\end{description}

\section{Dominant energy condition}\label{sect:dec}
In this section we describe some sufficient (but far from necessary) 
conditions on the
Lagrangian function $L$ of the form considered in the previous
section that guarantees the associated Einstein-Hilbert stress-energy
tensor satisfies the dominant energy condition. In isolated cases
(fluids and elasticity with the space-like simultaneous surface
assumption, or exactly the classical Skyrme and Born-Infeld models
\cite{Gibbon2003}) the results are known before. The novel
contribution in this paper is given below in Proposition
\ref{prop:rankcondition}. It can be used in a unified algebraic 
framework applicable
to all theory of maps described by an action of the form
\eref{eq:genlag}, making unnecessary the \emph{ad hoc} computations
through explicitly evaluating the eigenvalues used in e.g.\
\cite{Gibbon2003}. It is worth remarking that those types of 
computations rely on a
genericity argument to diagonalize a positive semidefinite quadratic
form relative to a Minkowski metric, a procedure which cannot be
carried out when $D^\phi$ admits a null eigenvector. By working
geometrically and tensorially, the computations described below avoids
this difficulty.

We start by reviewing some definitions. Recall that the
Einstein-Hilbert stress-energy tensor $T\in\Gamma(T^0_2M)$ for a 
Lagrangian field theory is given by a
variational derivative for the Lagrangian \emph{density} relative to
the inverse metric, 
\begin{equation}\label{eq:def:setensor}  T\sqrt{|\det g|} :=
\frac{\delta [L \sqrt{|\det g|}]}{\delta
g^{-1}} = \left(\frac{\delta L}{\delta g^{-1}}- \frac12 L g
\right)\sqrt{|\det g|}~.
\end{equation}

\begin{defn}
The stress-energy tensor $T$ is said to obey the \emph{dominant energy
condition} at a point $p\in M$ if $\forall X\in T_pM$ such that
$g(X,X) < 0$, the following two conditions are satisfied
\numparts
\begin{eqnarray}
\label{eq:DECpastpointing} T(X,X)  > 0 \\
\label{eq:DECcausal} [ T\circ g^{-1} \circ T ](X,X) \leq 0
\end{eqnarray}
\endnumparts
unless $T$ vanishes identically.
\end{defn}
\begin{rmk}
The definition is equivalent to the classical statements (see,
e.g.~section 4.3 in \cite{HawEll1973} or chapter 9 of
\cite{Wald1984}) of the dominant energy condition.
Observe that \eref{eq:DECcausal} gives that the vector $g^{-1}\circ
T\circ X$ is a causal vector for any time-like vector $X$, and
\eref{eq:DECpastpointing} gives that the vector $g^{-1}\circ
T\circ X$ has opposite time-orientation as the time-like vector $X$.
\end{rmk}

The set of future-pointing time-like vectors form a convex
cone; hence we have the following technical lemma, applicable to all
Lagrangian field theories, not just those described in Section
\ref{sect:lag}.
\begin{lem}\label{lem:convex}
Let $F = F(x_1,\ldots,x_k)$ be a continuously differentiable function of 
$k$ real variables. Assume that $F$ is concave, $F(0) \geq 0$, and that 
$\partial_iF \geq 0$ for each $1\leq i \leq k$. 
Let $L_i$, $1\leq i \leq k$ denote a collection of Lagrangian
functions, and let $T_i$ denote their corresponding stress-energy
tensors. Suppose $T_i$ each separately obeys the dominant energy 
condition, or, equivalently,
the vectors $Y_i = g^{-1}\circ T_i\circ X$ are all past-causal for any
fixed future time-like $X$. Then the stress-energy tensor $T$ for the
Lagrangian formed by 
$L = F(L_1,\ldots,L_k)$ also obeys the dominant energy condition.
\end{lem}
\begin{proof}
The stress-energy tensor $T$ can be written, using
\eref{eq:def:setensor}, as
\[ T = \sum_{i = 1}^{k} \partial_i F \cdot \frac{\delta
L_i}{\delta g^{-1}} - \frac12 F g = \sum_{i=1}^{k}
\partial_i F \cdot T_i - \frac12 (F - \sum_{i=1}^{k}\partial_i
F\cdot L_i) g~.\]
Now considering $g^{-1}\circ T\circ X$, the first term in the above
expression contributes $\sum \partial_i F \cdot Y_i$. By assumption, 
this is a convex combination of past-causal
vectors, and hence is past-causal. For the second term, since $g^{-1}\circ g\circ X = X$, to
show that it is also past-causal it suffices to show that
\[ F \geq \sum_{i=1}^{k}\partial_i F\cdot L_i~. \]
But this follows from the fact that $F$ is concave and $F(0) \geq 0$.
\end{proof}
\begin{rmk}\label{rmk:convex}
That $F$ is required
to have nonnegative partial derivatives represents the fact that each of
the $L_i$'s contribute nonnegatively to the energy. That $F(0) \geq 0$
states that there is no negative vacuum energy. Both conditions are
therefore natural and necessary for the total Lagrangian $L$ to have
positive energy density, if the $L_i$'s are taken to be independent. 
The concavity condition is technical. It
appears naturally in the proof, but can potentially be relaxed if 
more is assumed on the individual $L_i$'s. 
\end{rmk}

We will apply Lemma \ref{lem:convex} to the following proposition,
which is the main computational result of this section. Observe also
that for $L = s(\phi)$, its corresponding stress-energy tensor is $T =
-\frac12 s\cdot g$, and by the assumption on the positivity of $s$,
obeys the dominant energy condition. 
\begin{prop}\label{prop:rankcondition}
For $L = \sigma_j(D^\phi)$, $T$ obeys the dominant energy condition. 
Furthermore, $T = 0$ at a point $p$ if and only if $j > \mathrm{rank}(d\phi |_p)$.
\end{prop}

Before giving the proof, we need to review some linear algebra.
Consider a real vector space $V$. Let $A$ be a linear transformation
on
$V$. Then $A$ naturally extends to a linear transformation, which we
denote $A^{\sharp j}$, on
$\Lambda^j(V)$, the space of alternating $j$-vectors over $V$. 
A classical result in linear algebra is that $\sigma_j(A)$ is 
proportional to $\tr_{\Lambda^j(V)}A^{\sharp j}$. Now, letting $V = T_pM$ and
$A = D^\phi = g^{-1}\circ \phi^*h$, we observe that
\[ (D^\phi)^{\sharp j} = (g^{-1})^{\sharp j} \circ \phi^* (h^{\sharp
j})~, \]
or, to put it in words, $(D^\phi)^{\sharp j}$ is obtained from first
taking the induced metric $h^{\sharp j}$ on alternating $j$-vectors in
$T_{\phi(p)}N$, pulling it back via $\phi$, and composing it with the
induced metric $(g^{-1})^{\sharp j}$ for the alternating $j$-forms. In
index notation, this can be written as
\[
[(D^\phi)^{\sharp j}]^{b_1\ldots b_j}_{a_1\ldots a_j} =
g^{b_1c_1}\cdots g^{b_jc_j} (\phi^*h)_{a_1 [c_1|} 
(\phi^*h)_{a_2 |c_2|}\cdots (\phi^*h)_{a_{j-1} |c_{j-1}|}
(\phi^*h)_{a_j |c_j]}
\]
where the bracket notation in the indices denotes full
anti-symmetrization of the $\{c_1,\ldots,c_j\}$ indices. For a
Lagrangian proportional to $\sigma_j$,  we can assume
\begin{equation}\label{eq:lagrangiantensorly}
L = [(D^\phi)^{\sharp j}]^{b_1\ldots b_j}_{a_1\ldots a_j} =
g^{a_1[c_1|}\cdots g^{a_j|c_j]} 
(\phi^*h)_{a_1 c_1} \cdots (\phi^*h)_{a_jc_j}~.
\end{equation}
It is simple to check, using $(D^\phi) = \mathop{diag}(-1,1,1,\ldots)$
that the above expression has the correct sign: that $L$ defined thus
is a positive multiple of $\sigma_j$.

One can also arrive at \eref{eq:lagrangiantensorly} purely from a
linear algebra point of view. Let $p_j$ be the power sum
\[ p_j(\{\lambda_1,\ldots,\lambda_k\}) = \sum_{i=1}^{k} \lambda_i^j~.
\]
Recall Newton's identity
\[ j\cdot \sigma_j = \sum_{i=1}^j (-1)^{i-1}\sigma_{j-i}p_i \]
which allows us to express $\sigma_j$ as a rational polynomial in $p_i$'s.
Now, by definition, it is clear that
\[ p_j(D^\phi) = \tr [(D^\phi)^j] \]
where $(D^\phi)^j$ is the $j$-fold composition of $D^\phi$. It is easy
to check then, for some $E$
\[ \sigma_j = g^{a_1b_1}\cdots g^{a_jb_j} E_{b_1\ldots b_j}^{c_1\ldots c_j}
(\phi^*h)_{a_1c_1}\cdots (\phi^*h)_{a_jc_j}~. \]
Newton's identity reduces to a generating condition for $E$ based on
the Kronecker $\delta$ symbols,
\begin{eqnarray*}
E_b^c = \delta_b^c~,\\
j E_{b_1\ldots b_j}^{c_1\ldots c_j} = \sum_{i = 1}^j
(-1)^{i-1}E_{b_1\ldots b_{j-i}}^{c_1\ldots c_{j-i}}
\delta_{b_{j-i+1}}^{c_{j-i}}\delta_{b_{j-i+2}}^{c_{j-i+1}}\cdots\delta_{b_j}^{c_{j-i+1}}~.
\end{eqnarray*}
A direct computation which we omit here shows that then in fact the
invariant $E_{b_1\ldots b_j}^{c_1\ldots c_j}$ is a positive rational
multiple of the generalized Kronecker symbol
$\delta_{b_1\ldots b_k}^{c_1\ldots c_j}$, from which we recover
\eref{eq:lagrangiantensorly}.

\begin{proof}[Proof of Proposition \ref{prop:rankcondition}]
We need to show that $g^{-1}\circ
T \circ X$ is past causal for any future time-like vector $X$. Since
$T$ is tensorial, we can assume $X$ has unit length. Using the
expression for (a positive scalar multiple of) $\sigma_j$ given in
\eref{eq:lagrangiantensorly}, we can write $T(X,\cdot)$ for $L =
\sigma_j$ in index notation:
\begin{equation}\label{eq:setensortensorly}
T_{ab}X^b = j X^{[b|}g^{a_2|c_2|}\cdots g^{a_j|c_j]}(\phi^*h)_{a
b} \cdots (\phi^*h)_{a_jc_j} - \frac12 g_{ab}X^b L
\end{equation}

Take an orthonormal basis for $T_pM$ relative to $g$. Since we
assumed $X$ unit, let $e_0 = X$ and $\{e_i\}_{1\leq i\leq m}$ all
space-like. We can take $j \leq m+1$ as otherwise $T$ is identically
0. Then we notice that a basis for $\Lambda^j(T_pM)$ is given by
\[ \{ e_0 \wedge e_{\alpha_1}\wedge \cdots \wedge e_{\alpha_{j-1}}
\}_{1\leq \alpha_1 < \cdots < \alpha_{j-1}\leq m} \cup \{
e_{\alpha_1}\wedge \cdots \wedge e_{\alpha_j}\}_{1\leq \alpha_1 <
\cdots < \alpha_j\leq m}~. \]
We write the first set as $\Lambda^j_\perp$ and the second set as
$\Lambda_\parallel^j$. Using the normalization that $v\wedge w =
v\otimes w - w\otimes v$, we find that each of the element in
$\Lambda^j_\perp$ has norm $-j!$ while the elements in
$\Lambda^j_\parallel$ has norm $j!$.

To show that $T(X,X) > 0$, we observe that under the
expansion \eref{eq:setensortensorly}, the first term corresponds to
\[ \sum_{\omega\in \Lambda_\perp^j} \phi^* ( h^{\sharp j}
)(\omega,\omega)~, \]
while the second term corresponds to
\[ \frac{1}{2} \left( - \sum_{\omega\in\Lambda_\perp^j} \phi^* (
h^{\sharp j})(\omega,\omega)  +
\sum_{\omega\in\Lambda_\parallel^j}
\phi^* (h^{\sharp j})(\omega,\omega)\right)~. \]
So summing them gives
\[ \frac{1}{2} \left( \sum_{\omega\in\Lambda_\perp^j} \phi^* (
h^{\sharp j})(\omega,\omega) +
\sum_{\omega\in\Lambda_\parallel^j}
\phi^* (h^{\sharp j})(\omega,\omega)\right) \]
which is non-negative by the fact that $\phi^* (h^{\sharp j})$ is a
positive semi-definite quadratic form on $\Lambda^j(T_pM)$.
Furthermore, observe that since
$\Lambda^j_\parallel\cup\Lambda^j_\perp$ is a basis, its push-forward
$\phi_*\Lambda^j_\parallel\cup \phi_*\Lambda^j_\perp$ spans
$\Lambda^j(\phi_* T_pM) \subset \Lambda^j(T_{\phi(p)}N)$. Thus by the
fact that $h$ (and hence the induced metric $h^{\sharp j}$) is
positive definite, we conclude that when $T(X,X) = 0$, necessarily 
$\Lambda^j(\phi_* T_pM) = \{ 0\}$. This proves the assertion 
that $T$ vanishes only when $j > \mathrm{rank}(d\phi)$.

To show \eref{eq:DECcausal}, we observe that
\[
X^aT_{ac}g^{cd}T_{db}X^b = - T(X,X)^2 + \sum_{i = 1}^m T(X,e_i)^2~.
\]
The first thing to note is that $T(X,e_i)$ does not have any 
contribution from the second term in \eref{eq:setensortensorly}. For
the first term, a quick computation shows that $T(X,ei)$ corresponds
to 
\[ \sum_{\eta \in \Lambda^{j-1}_\parallel} \phi^*(h^{\sharp
j})(e_0\wedge \eta,
e_i\wedge\eta) \]
so 
\begin{eqnarray*}
\fl |\sum_{i=1}^m T(X,e_i)^2 | & \leq& ( \sum |T(X,e_i)| )^2 \\
& \leq &( \sum_{i = 1}^m \sum_{\eta\in \Lambda^{j-1}_\parallel,
e_i\wedge \eta \neq 0}
|\phi^*(h^{\sharp j})(e_0\wedge \eta, e_i\wedge\eta)|)^2 \\
& \leq& \frac14(\sum_{\eta\in \Lambda^{j-1}_\parallel}\phi^*(h^{\sharp
j})(e_0\wedge \eta,e_0\wedge \eta) + \sum_{i=1}^m\phi^*(h^{\sharp
j})(e_i\wedge \eta,e_i\wedge \eta))^2 \\
& = &\frac14(\sum_{\eta\in
\Lambda^{j-1}_\parallel}\sum_{i=0}^m\phi^*(h^{\sharp
j})(e_i\wedge \eta,e_i\wedge \eta))^2 \\
& = &T(X,X)^2
\end{eqnarray*}
and therefore \eref{eq:DECcausal} is satisfied.
\end{proof}

As an immediate application of Lemma \ref{lem:convex} and Proposition
\ref{prop:rankcondition}, we have that the Skyrme model and the
Born-Infeld model described in Section \ref{sect:lag} obey the
dominant energy condition: it suffices to check that
$F_{\mathrm{Skyrme}}(s,\sigma_1,\sigma_2) = s + \sigma_1 + \sigma_2$
and $F_{\mathrm{BI}}(\sigma_1, \ldots, \sigma_{m+1}) =
\sqrt{\sum_1^{m+1} b^{m+1-j} \sigma_j} - \sqrt{b^{m+1}}$ are concave,
satisfies $F(0) \geq 0$, and has positive partial derivatives,
conditions which are easily seen to hold. Therefore we recover the
result of \cite{Gibbon2003} in dimension $m = n = 3$, and also extend 
it to arbitrary dimensions $m,n$. In fact, we have

\begin{thm}\label{thm:lagsuffdec}
Given a Lagrangian theory of maps with action given by
\eref{eq:genlag}, a sufficient condition for its 
Einstein-Hilbert stress-energy tensor to obey the dominant energy
condition is that the Lagrangian function $L$ in \eref{eq:genlag} 
be continuously differentiable with nonnegative partial derivatives on
its arguments, be concave, and satisfy $L(0) \geq 0$. 
\end{thm}
\begin{rmk}
We stress again here that our notion of Lagrangian theory of maps, as
discussed in Section \ref{sect:lag}, is very general, and includes the
general forms of Lagrangians used in relativistic elasticity and
relativistic fluids. Therefore it is remarkable that for these theory
of maps, the assumption of a time-like particle world-line is
inconsequential insofar as the dominant energy condition is concerned:
in particular a tachyonic fluid will still obey the dominant energy
condition, in direct contradiction to the intuition often presented as
the justification for the dominant energy condition. For a concrete
example, consider the fluid Lagrangian with some large constant $b$ on
Minkowski space,
\[ L = \sqrt{b + \sigma_3(D^\phi)}~. \]
By Theorem \ref{thm:lagsuffdec} it obeys the dominant energy
condition; locally, however, the ``fluid'' moving with a constant
``velocity'' larger than that of gravity is a solution, e.g.\
$\phi:\mathbb{R}^{1+3}\to\mathbb{R}^3$ where $\phi(t,x,y,z) = (t,x,y)$
is a solution to the equations of motion. 
As we will see in Section \ref{sect:ex} for the Skyrme
model, however, there is an instability (i.e.\ a violation of
hyperbolicity) associated to these tachyonic matter. 
\end{rmk}

\section{Regular hyperbolicity and the canonical
stress}\label{sect:rh}
The question of hyperbolicity for a system of partial differential
equations is generally synonymous with whether the system admits a
locally well-posed Cauchy problem for \emph{all} smooth initial data. 
In the context of linear systems of
constant coefficients, it is a well-known theorem \cite[Ch.\ 5]{Horman1976} 
that a necessary and
sufficient condition for hyperbolicity of a system $P(\partial)\phi = 0$ 
is the hyperbolicity for its polynomial \emph{symbol} $P(\xi)$. 
(That the Cauchy problem is well-posed for \emph{all} smooth initial
data is crucial for the necessity. See Remark \ref{rmk:ultrahyper}
below.) Two difficulties arise when trying to apply this theorem to
non-linear or variable coefficient systems. First is the general
difficulty of checking the hyperbolicity for a symbol, which requires
computing the singular locus of $P(\xi)^{-1}$. Second is the technical
difficulty that concrete quantitative estimates are required for
iteration arguments used in these type of problems. 

Many stronger notions of hyperbolicity exist in the mathematical
literature (see \cite{CouHil1962, Lax2006} for some samplers); the
common theme to all is that they provide \emph{sufficient} conditions
for local well-posedness of the non-characteristic Cauchy problem. In
general, however, these conditions are not \emph{necessary}: that
Leray hyperbolicity \cite{Leray1952} is a subset of symmetric
hyperbolicity \cite{Friedr1954} is well-known; that the
Maxwell-Born-Infeld model can only be seen as symmetric hyperbolic
after an augmentation procedure \cite{Serre2004,Brenie2004} strongly
suggests that the notion of symmetric hyperbolicity does not directly
capture all hyperbolic systems.
(It is perhaps interesting to note that the hyperbolicity of 
Maxwell-Born-Infeld system can be directly treated within the regular
hyperbolicity framework \cite{Speck-p2010}.) In this section we will
describe the regular hyperbolicity framework, which provides another
sufficient condition for local well-posedness\footnote{That this
framework is not better known in the community is perhaps due to the
dense mathematical language in Christodoulou's monograph \cite{Christ2000}. 
The author hopes to provide here, by way of a summary of
the main ideas behind the theory, an advertisement for this simple yet
powerful technique that is well-adapted for use in Lagrangian field
theories. In a forthcoming paper \cite{SpeWon-pPP}, Jared Speck and
the author will give a more accurate and detailed introduction to
regular hyperbolicity and fill in some material omitted in
\cite{Christ2000}.}. To guarantee local
well-posedness, the general technique common to these methods is that
of the $L^2$ energy estimate.

\subsection{Energy estimates}
The energy estimates that we will need are estimates of the following
form: assume for the current exposition that $M$ is diffeomorphic 
to $\mathbb{R}^{1+m}$, and let
$H^k$ denote the $L^2$ Sobolev space with $k$ derivatives on
constant time slices diffeomorphic to $\mathbb{R}^m$. The energy
estimates are inequalities 
\begin{equation}\label{eq:sobolev}
\| \phi \|_{H^k}(t) \leq \|\phi\|_{H^k}(0) e^{Ct} \quad \forall 0 < t
< T, C = C(T, \|\phi\|_{H^k}(0))~. 
\end{equation}
(That we are only interested in such $L^2$ energy estimates is because
we look for estimates generally applicable to hyperbolic equations. In
particular, we require the estimates to hold for the linear wave
equation. By a theorem of Littman \cite{Littma1963}, inequalities of
the form \eref{eq:sobolev} for the wave equation can only hold in $L^2$ 
based function spaces.) Inequalities of the form \eref{eq:sobolev}
guarantee local existence up to time $T$. 

To obtain such energy estimates, the usual method is via the divergence
theorem. Let $J = J(t,x,\phi,\ldots,\partial^k\phi)$ be a vector
field, depending on up to $k$ derivatives of $\phi$ which we assume
solves some partial differential equation.
The divergence theorem states that, for $n = \partial_t$ the
normal vector field to the slices
\[ -\int J(0,x)\cdot n dx + \int J(t,x)\cdot n dx =
\int_0^t\int \mathop{div}(J) dx dt~. \]
The inequality \eref{eq:sobolev} would hold after one application of
Gronwall's lemma, if we can guarantee the following
\numparts
\begin{eqnarray}
\label{eq:positivity}C^{-1} \|\phi\|_{H^k}(t) & \leq \int J(t,x)\cdot n dx \leq
& C\|\phi\|_{H^k}(t) \quad \mbox{ for some } C \geq 1~,\\
\label{eq:noetherineq}\int\mathop{div}(J)(t) dx & \leq C' \|\phi\|_{H^k}(t) & \mbox{ for
some } C' > 0~.
\end{eqnarray}
\endnumparts
A $J$ verifying \eref{eq:positivity} and  \eref{eq:noetherineq} will
be called a \emph{compatible energy current}. (In \cite{Christ2000},
Christodoulou use the term ``compatible current'' to refer to any $J$
verifying \eref{eq:noetherineq}. The use of the word ``energy'' here is
meant to reflect the imposition of the additional positivity condition
\eref{eq:positivity}.)

For \eref{eq:noetherineq} to hold, it is necessary that we apply the
equation: by the chain rule, the divergence of the vector field $J$
which depends on $k$ derivatives of $\phi$ will depend on $k+1$
derivatives of $\phi$; it will be absurd to be able to control its
integral by something depending on fewer derivatives. But by suitably 
applying the equations of motion for $\phi$, we can
convert top-order derivatives to lower-order ones and satisfy the
inequality. 

\subsection{Regular hyperbolicity}
Christodoulou provided in his monograph \cite{Christ2000} a robust and
geometric way for obtaining compatible energy currents for any given 
Lagrangian theory of maps, using techniques similar to those used by
Hughes, Kato and Marsden in a non-geometric context \cite{HuKaMa1976}.
In this subsection, we quickly review the key points of the theory.
The main results obtained by Christodoulou are that
\begin{enumerate}
\item Vector fields satisfying \eref{eq:noetherineq}, up to a total
divergence term, and up to Noether currents coming from symmetries of
the target manifold, \emph{generically} arise from contractions of 
arbitrary vector
fields against what is called the \emph{canonical stress} tensor,
which can be obtained algorithmically from the equation of motion. 
\item The existence of vector fields satisfying \eref{eq:positivity}
depends only on the causal properties of the canonical stress tensor. 
\end{enumerate}

Let us begin by summarising the construction of the canonical stress
tensor. Let $\psi$ be a function taking values in $\mathbb{R}^n$. We
will use index notation where components of $\mathbb{R}^n$ are
indicated by capital Latin letters. Assume $\psi$ solves a system of
second order partial differential equations of the form
\begin{equation}\label{eq:modelequation}
m^{ab}_{AB}\partial^2_{ab} \psi^B = F_A(\psi, \partial\psi)
\end{equation}
where lowercase Latin letters denote indices on the space-time
manifold, and $m^{ab}_{AB}$ is some field of coefficients, symmetric
in $a,b$, and symmetric in $A,B$. Consider the tensor field defined by
\begin{equation}
Z^{ab}_{AB}|^c_d := m^{ab}_{AB}\delta^c_d - m^{cb}_{AB}\delta^a_d -
m^{ac}_{AB}\delta^b_d~,
\end{equation}
it is a direct computation to show that the tensor field
\begin{equation}
Q[\psi]^c_d := - Z^{ab}_{AB}|^c_d \partial_a\psi^A \partial_b\psi^B 
\end{equation}
has the property
\begin{equation}
\partial_cQ[\psi]^c_d = - (\partial_c
Z^{ab}_{AB}|^c_d)\partial_a\psi^A\partial_b\psi^B +
2(\partial_d\phi^B)F_B(\psi,\partial\psi)~.
\end{equation}
Therefore, for any vector field $X$, the vector field $Q^c_dX^d$ has
the property that it depends on 1 derivative of $\psi$, and so does
its divergence. This tensor $Q^c_d[\psi]$ is the \emph{canonical stress
tensor} associated to the solution $\psi$ of the equation 
\eref{eq:modelequation}. Notice it
only depends on $\partial\psi$ and the coefficients for the highest 
order derivative term, $m^{ab}_{AB}$

In general, for a Lagrangian theory of maps described in Section
\ref{sect:lag}, the Euler-Lagrange equation will take the quasilinear
form
\[ m^{ab}_{AB}(x,\phi,\partial\phi)\nabla^2_{ab}\phi^B =
F_A(x,\phi,\partial\phi)~.\]
By commuting in further derivatives, we see that the partial
derivatives $\psi^A = \nabla^\alpha\phi^A$ all satisfy equations of
the form \eref{eq:modelequation} \emph{with the same $m^{ab}_{AB}$}.
Therefore, with judicious applications of Sobolev embedding theorem,
we see that for any fixed $X$, the vector field 
\[ J^c := X^c|\phi|^2 + \sum_{|\alpha| \leq k}
Q[\nabla^\alpha\phi]^c_dX^d \]
satisfies \eref{eq:noetherineq}. Note that thus far the construction
of $J$ is an algebraic and algorithmic statement about Lagrangian
theory of maps. 

To obtain hyperbolicity, it is necessary to also satisfy
\eref{eq:positivity}. It was shown by Christodoulou that given a
foliation $\Sigma_t$ of $M$ by level sets of a function $t$, that a
sufficient condition for $J$ (defined via the vector field $X$) to
satisfy \eref{eq:positivity} is for the following to
hold
\begin{enumerate}
\item $m^{ab}_{AB}(dt)_a(dt)_b$ is a negative definite matrix;
\item $m^{ab}_{AB}\xi_a\xi_b$ is a positive definite matrix for all
non-zero one-forms satisfying $\xi(X) = 0$. 
\end{enumerate}
(Please refer to \cite{Christ2000} for proof.)
We shall call functions $t$ that satisfy the first property \emph{time
functions} and vector fields $X$ that satisfy the second property
\emph{observer fields}. For semilinear equations where $m^{ab}_{AB} =
g^{ab}h_{AB}$, these notions agree with the that of time
functions and timelike observer fields in general relativity. In the
quasilinear case, the time functions and observer fields form a
replacement for the usual causal structure of the Lorentzian metric
for governing the kinematics in perturbative analysis of solutions.
One can also analogously define the notion of global hyperbolicity,
domain of dependence, and maximal development of initial data relative
to this replacement causal structure. (This formulation is, in
particular, used in Christodoulou's seminal work on shock
instabilities of the Euler equation \cite{Christ2007}.)
In \cite{Christ2000} it is claimed without proof that the existence of a
time function and an observer field is sufficient to guarantee the
local well-posedness of the Cauchy problem with data prescribed on a
level surface of $t$, a proof will be supplied in \cite{SpeWon-pPP}; 
for the non-geometric scenario working over a
fixed coordinate system in Minkowski space, with $X = \partial_t$, a
proof is available in \cite{HuKaMa1976}. 

\subsection{Breakdowns of hyperbolicity}
When the existence of time functions and observer fields fail, the
regular hyperbolicity of the system breaks down. This in particular
implies the non-existence of general energy estimates by the work of
Christodoulou \cite{Christ2000}, and hence the impossibility of
applying the usual iteration method to obtain local existence and
uniqueness of solutions to the Cauchy problem. While this does not
prove the lack of hyperbolicity for those systems, here we provide
some heuristic arguments as to why the lack of regular hyperbolicity
is indicative of a lack of local well-posedness\footnote{We should
recall at this stage that we consider well-posedness in the sense of
Hadamard \cite{Hadama1923}: the existence of suitably regular
solutions, the uniqueness of said solutions, and the continuous
dependence of the solutions on given initial data.}.

To illustrate the different modes of breakdown, we first consider the
linear equations given by
\[ \epsilon_0 u_{tt} + \epsilon_1 u_{11} + \epsilon_2 u_{22} +
\epsilon_3 u_{33} \]
for some scalar $u$ on $\mathbb{R}^4$, with $\epsilon_*$ taking 
values $\pm 1$. The coefficients $m^{ab}_{AB}$ for this equation take
values in $1\times 1$ matrices, i.e.\ scalars, and we have that 
\[ m^{ab} = \mathop{diag}(\epsilon_0, \epsilon_1,
\epsilon_2,\epsilon_3)~.\]
The usual case of the wave equation is given by $-\epsilon_0 =
\epsilon_1 = \epsilon_2=\epsilon_3 = 1$. It is clear that then for a
given covector 
\[ \xi = \xi_0 dt + \xi_1 dx^1 + \xi_2 dx^2 + \xi_3dx^3~, \]
the requirement for $m^{ab}\xi_{a}\xi_b$ to be negative definite (that
is, negative as a scalar) is the usual condition that $\xi$ is
time-like with respect to the metric $m^{ab}$: $\xi_0^2 > \xi_1^2 +
\xi_2^2 + \xi_3^2$. Similarly, for $m^{ab}\xi_a\xi_b$ to be negative,
we just reverse the preceding inequality. Hence the wave equation is
hyperbolic if we choose $t$ to be a time-like (in the usual sense)
foliation and $X$ to be any time-like (again in the usual sense)
vector. 

The standard example in which one cannot construct any time-function
is the case of Laplace's equation, where $\epsilon_* = 1$, and for any
covector $m^{ab}\xi_a\xi_b$ is positive. For elliptic problems, it is
well-known \cite{Hadama1927, Tarkha1995} that the Cauchy problem is
ill-posed: there cannot be continuous dependence on initial data. 
By comparison, we shall say
that a Lagrangian theory of maps have an \emph{elliptic type
breakdown} of hyperbolicity if one cannot construct any
time-function even locally. Note that elliptic type breakdowns are not
the same as the equations forming a \emph{bona fide} regularly elliptic
system, which requires the ``Legendre-Hadamard condition''
\cite{Christ2000} that for all covectors $\xi$,
$m^{ab}_{AB}\xi_a\xi_b$ be positive definite matrices. What we call
elliptic type break down only suffices that none of those matrices 
be negative definite, and in particular mixed signatures will imply
break down.

The case where one can construct a time-function but not any observer
fields we refer to as \emph{ultrahyperbolic type breakdown}; the name
is taken from the canonical example of the ultrahyperbolic equation
\[ -u_{tt} - u_{11} + u_{22} + u_{33} = 0 \]
where any foliation with normal covector $\xi$ satisfying $\xi_0^2 +
\xi_1^2 > \xi_2^2 + \xi_3^2$ contributes a time-function, but the
trace of $m^{ab}$ to any three-plane is indefinite. The ultrahyperbolic
equation have infinite speeds of propagation \cite{John1955}, and can
be checked by the theorem alluded to in the beginning of this section
to have a non-hyperbolic polynomial symbol, and hence cannot admit
well-posed Cauchy problems. (The instability of  ultrahyperbolic
equations has also been considered on physical grounds in the
literature, see e.g.\ \cite{Dorlin1970}.)
\begin{rmk}\label{rmk:ultrahyper}
The ultrahyperbolic equation illustrates an important connection
between hyperbolicity and finite speed of propagation. As mentioned,
the ultrahyperbolic equation does not admit well-posed Cauchy problem
for \emph{all} smooth initial data. It was however pointed out by
Craig and Weinstein \cite{CraWei2009} that, if one were willing to impose
``non-local constraints'' (in their case a correlation on the
admissible space-time frequencies of the waves), the Cauchy problem
can be well-posed in the restricted class. One should think of the
non-local constraints as circumventing the instabilities caused by
infinite speeds of propagation. For semilinear equations, it may be
possible to impose such constraints \emph{a priori} and globally
(provided these additional constraints are compatible with the 
equations of motions; the constraint given in \cite[\S 2]{CraWei2009}
is not preserved if one were to modify the ultrahyperbolic equation by
a power nonlinearity $|u|^pu$); the situation for quasilinear equations is much
less clear. 
\end{rmk}

\subsection{Applying to Lagrangian theory of maps}
A direct computation shows that the coefficient tensor $m^{ab}_{AB}$
can be obtained as the second variational derivative of the Lagrangian
function relative to the field velocity. That is:
\begin{equation}\label{eq:defprincipal}
m^{ab}_{AB}(x,\phi,\partial\phi) = \frac{\delta^2}{\delta
\partial_a\phi^{A}~\delta\partial_b\phi^{B}}L(x,\phi,\partial\phi)~.
\end{equation}
Noting that the positive (negative) definite matrices form a convex
cone inside the space of matrices, we see that if $t$ and $X$ are
simultaneously time functions and observer fields for a collection of
Lagrangian functions $L_i$, then they will also form a pair of time
function and observer field for any convex sum of the $L_i$'s. 

Turning now to the individual cases $L = \sigma_j(D^\phi)$ as
described in Section \ref{sect:lag}, we can again use
\eref{eq:lagrangiantensorly} to facilitate computations. Directly we
get
\begin{eqnarray*}
\fl \frac{\delta^2L}{\delta \partial_d\phi^D~\delta\partial_b\phi^B}
= 2j g^{a_1[c_1|}\cdots g^{a_j|c_j]} (\phi^*h)_{a_1 c_1} \cdots 
(\phi^*h)_{a_{j-2}c_{j-2}}\\
\cdot \left[ (\phi^*h)_{a_{j-1} c_{j-1}}
\delta_{a_j}^b\delta_{c_j}^dh_{BD} + 2(j-1)
\delta_{(c_{j-1}}^b\partial_{a_{j-1})}\phi^A
\partial_{(a_j}\phi^{A'}\delta_{c_j)}^dh_{AB}h_{A'D}
\right] \end{eqnarray*}
Now consider an arbitrary time-like unit covector $\xi$, and let $\hat{g}$
be the restriction of $g$ to the orthogonal complement of $\xi$ (in
particular it is a positive semidefinite bilinear form), we see that
we can write $m^{bd}_{BD}\xi_b\xi_d$ in the following way
\begin{eqnarray}
\fl \nonumber m^{bd}_{BD}\xi_b\xi_d = -2 \hat{g}^{a_1[c_1|}\cdots \hat{g}^{a_{j-1}|c_{j-1}]}
(\phi^*h)_{a_1 c_1} \cdots (\phi^*h)_{a_{j-2}c_{j-2}} \\
\label{eq:existtimefunction} \cdot \left[
(\phi^*h)_{a_{j-1} c_{j-1}} h_{BD} - 
\partial_{c_{j-1}}\phi^A\partial_{a_{j-1}}\phi^Ch_{AB}h_{CD}
\right]
\end{eqnarray}
after carefully counting the various permutations of indices leaving
the expression non-zero. Using a Cauchy inequality on the term now in
the brackets, we see that it is positive semi-definite; it will be
positive definite whenever $d\hat{\phi}|_{p\in M}: T_pM\supset
\xi^\perp \to T_{\phi(p)}N$ is has rank at least 2 (that is, the 
tangent space map $d\phi$ restricted to the vectors in the kernel of 
$\xi$ has rank at least 2). Therefore, we have the following theorem.

\begin{thm}\label{thm:existtimefn}
Let the Lagrangian $L$ be a convex linear combination of
$\sigma_j(D^\phi)$ and $s(\phi)$. And let $t$ be an arbitrary time
function relative to the underlying metric $g_{ab}$. Then
\begin{enumerate}
\item if the coefficient for $\sigma_1$ is non-zero, then $t$ is a
time function for the equations of motion;
\item if $d\phi$ restricted to the level sets of $t$ has rank $\geq
\max(j-1,2)$, then $t$ is a time function for the equations of motion;
\item failing both of the above, $t$ is borderline degenerate, but
the corresponding $m^{ab}_{AB}(dt)_a(dt)_b$ is positive semi-definite. 
\end{enumerate} 
In particular, $L$ cannot have a \emph{bona fide} elliptic type
breakdown of hyperbolicity. 
\end{thm} 
Noting that the domain of dependence in theory of hyperbolic equations
is determined precisely by the admissible time-functions
\cite{Leray1952, Christ2000}, we have the following interesting
corollary, which implies that a Lagrangian theory of maps with the
dominant energy condition (compare Theorem \ref{thm:lagsuffdec}), when
in fact hyperbolic, cannot propagate perturbations faster than the
speed of gravity. 
\begin{cor}
Let $L = L(s,\sigma_1,\ldots,\sigma_j)$ be a Lagrangian function
with the property that $L$ is non-decreasing and concave in its
arguments. Assume further that for some $k$, the partial derivative of
$L$ with respect to $\sigma_k$ is bounded from below, and that either
(a) $k = 1$ (b) admissible solutions have $\mathop{rank}(d\phi) \geq 
\max(k,3)$. 
Supposing the dynamics derived from $L$ is regularly hyperbolic, then
the domain of dependence for solutions must be strictly contained
within that of the linear wave equation on the space-time $(M,g)$.
\end{cor}
\begin{proof}
Using the energy estimate, it suffices to show that any
time function $t$ of the underlying space-time is a time function for
the Lagrangian theory (see \cite[\S 5.3]{Christ2000} for a
discussion). Using \eref{eq:defprincipal} we write
\[ m^{ab}_{AB} = \frac{\partial^2L}{\partial \sigma_i
\partial\sigma_j} \frac{\delta \sigma_i}{\delta
\partial_a\phi^A} \frac{\delta \sigma_j}{\delta
\partial_b\phi^B} + \frac{\partial L}{\partial \sigma_j}
\frac{\delta^2\sigma_j}{\delta \partial_a\phi^A~\delta
\partial_b\phi^B}~.\]
After contracting with $(dt)_a(dt)_b$, 
the first term is positive semi-definite using the concavity of the
Lagrangian functional. The second term is a convex sum of positive
semi-definite matrices, and by Theorem \ref{thm:existtimefn} and the
hypotheses, at least one of the terms is positive definite. Hence $t$
is a time function for the Lagrangian theory of maps described by $L$. 
\end{proof}
\begin{rmk}
We can compare the above corollary to the situation in isentropic
fluids. After choosing an equation of state, one can write the
pressure $p$ of the fluid as a function of the proper energy density
$\rho$. In this regime the well-known criterion for hyperbolicity is
that the speed of sound $\sqrt{dp/d\rho}$ is positive and real, while the
criterion for the fluid flow to be causal is that the speed of sound
is less than the speed of gravity (which we can choose to equal 1)
\cite{Tahvil1998}. The proper energy density is in fact the Lagrangian
function for a fluid. Writing $\tau = (\sigma_n(D^\phi))^{-1/2}$ for
the volume per particle, the pressure can be given \cite{Christ2000}
by $p = - \frac{d}{d\tau}(\tau \rho)$. Then with a bit of elementary
calculus, one sees that the causality conditions are equivalent to the
following statement on $L = L(\sigma_n(D^\phi))$:
\[ 0 > 2 \sigma_n L'' > - L'~. \]
The first inequality guarantees the speed of sound is smaller than
that of gravity, and it states precisely that $L$ is concave, agreeing
with the corollary above. The second inequality guarantees
hyperbolicity; compare with the discussion below. Note that a
necessary condition for the second inequality to hold is that $L$ is
increasing as a function of $\sigma_3$, which is also a part of the
hypothesis in the above corollary. 
\end{rmk}

The analogous statement for observer fields, however, does not hold.
Letting $\eta$ now be an arbitrary space-like unit covector, and let
$X$ be a unit time-like vector orthogonal to $\eta$. Denote now by
$\tilde{g}$ the restriction of $g$ to the orthogonal complement of
$\eta$ and the metric dual of $X$. We see
that the analogous statement to \eref{eq:existtimefunction} is the
following
\begin{eqnarray}
\fl \nonumber m^{bd}_{BD}\eta_b\eta_d = 2 \tilde{g}^{a_1[c_1|}\cdots
\tilde{g}^{a_{j-1}|c_{j-1}]}(\phi^*h)_{a_1 c_1} \cdots
(\phi^*h)_{a_{j-2}c_{j-2}} \\ 
\nonumber \qquad \cdot \left[(\phi^*h)_{a_{j-1} c_{j-1}} h_{BD} - 
\partial_{c_{j-1}}\phi^A\partial_{a_{j-1}}\phi^Ch_{AB}h_{CD}
\right] \\
\nonumber -\frac{2}{j-1}\tilde{g}^{a_1|c_1|}\cdots
\tilde{g}^{a_{j-2}|c_{j-2}]}(\phi^*h)_{a_1 c_1} \cdots
(\phi^*h)_{a_{j-2}c_{j-2}} \\
\qquad \cdot
\left[X\phi^AX\phi^Ch_{AC}h_{BD} -
X\phi^AX\phi^Ch_{AB}h_{CD} \right]\\
\nonumber -\frac{2j-4}{j-1}\tilde{g}^{a_1|c_1|}\cdots
\tilde{g}^{a_{j-2}|c_{j-2}]}(\phi^*h)_{a_1 c_1} \cdots
(\phi^*h)_{a_{j-3}c_{j-3}}\\
\nonumber \qquad X\phi^EX\phi^Fh_{EF}
\left[ (\phi^*h)_{a_{j-2} c_{j-2}} h_{BD} - 
\partial_{c_{j-2}}\phi^A\partial_{a_{j-2}}\phi^Ch_{AB}h_{CD}
\right]
\end{eqnarray}
where the first term on the right hand side is again positive semi-definite
by Cauchy's inequality, while the second and third terms are \emph{a
priori} negative semi-definite. In particular, comparing the first and
third terms, if $|X\phi|_h^2$ is large
compared to $\tilde{g}^{ac}(\phi^*h)_{ac}$, the bilinear form
$m^{bd}_{BD}\eta_b\eta_d$ can easily have negative eigenvalues; one
can compare this to the conclusion drawn in \cite{CruBel1994} about
instabilities of the Skyrme model. We are forced to conclude that
a general Lagrangian theory of maps is susceptible to ultrahyperbolic
type breakdowns of hyperbolicity. 

On the other hand, notice that in the case $X\phi = 0$, the form
$m^{bd}_{BD}\eta_b\eta_d$ is indeed positive semi-definite, and
furthermore positive definite if $d\phi$ satisfies an analogous
rank condition as in Theorem \ref{thm:existtimefn}. Also note
that in the case $j = 1$ (the semi-linear case), positive definiteness
of $m^{bd}_{BD}\eta_b\eta_d$ holds without need of any assumptions on
$X\phi$, as the problematic terms do not appear. 

\begin{rmk}
As seen in the discussions above, the term $\sigma_1$ in the
Lagrangian always introduce a factor that is regularly hyperbolic.
This in fact has a stabilizing effect on the hyperbolicity of the
field theory. An example will be given in Section \ref{sect:ex}, where
hyperbolicity of the Skyrme model persists into a regime where the
particle velocity (using a fluids interpretation) exceeds that of
speed of gravity. 
\end{rmk}

\subsection{Canonical stress versus Einstein-Hilbert
stress-energy}\label{sect:canstrvseh}
Thus far we have seen that for $\sigma_j$, the dominant energy
condition must hold; on the other hand, the computations above show
that the analogous statement for the canonical stress need not hold.
That is: for perturbations $\psi$ over a solution $\phi$ to the
equations of motions determined by $\sigma_j(D^\phi)$, the canonical
stress $Q[\psi]^c_d$ need \emph{not} have the property
``$Q[\psi]^c_dX^d\xi_c \geq 0$ for any time-like vector $X$ and 
time-like covector $\xi$ with $\xi(X) <0$''. However, it is also 
clear from the definitions that the canonical
stress tensor and the Einstein-Hilbert stress-energy tensor in fact
agree in the highest order derivative terms for the semilinear case $L
= \sigma_1(D^\phi)$, in which the equations of motion are always
regularly hyperbolic. What is, then, the difference between the two
stress tensors?

In fact we have the following:
\begin{prop}\label{prop:equivalence}
For $L = \sigma_j(D^\phi)$, the canonical stress tensor $Q[\phi]$
\emph{corresponding to the solution itself} agrees with $g^{-1}\circ
T$, the Einstein-Hilbert stress-energy tensor.
\end{prop} 
The proof follows from a direction computation of $Q[\phi]^c_dX^d$
using the formulae given in \eref{eq:defprincipal} \emph{et seq.}, and
comparison against \eref{eq:setensortensorly}. We omit the details
here. This proposition shows that for Lagrangians composed of convex
linear combinations of the $\sigma_j$'s, the dominant energy condition 
is sufficient to guarantee the existence of a compatible energy
current that controls the solution itself, but not its higher
derivatives. This is, of course, unsurprising, as the Einstein-Hilbert
stress-energy tensor is divergence free by definition, and hence for
any time-like vector field $X$, $g^{-1}\circ T\circ X$ is a compatible
current. 

It is tempting to propose that, in view of Proposition
\ref{prop:equivalence}, regular hyperbolicity implies the dominant
energy condition. This is, however, not necessarily the case in
general. Returning to \eref{eq:positivity} of the fundamental energy
estimate, we see that it only suffices that the compatible current is
comparable with the Sobolev norms \emph{in an integrated sense}, and
indeed, the conditions on regular hyperbolicity encoded in positivity
and negativity conditions on $m^{ab}_{AB}$ only guarantees that much.
By way of an example, we consider a linear equation on
$\mathbb{R}^{1+2}$ with a two-dimensional target manifold. Consider
the following set of matrices
\numparts
\begin{eqnarray}
\tilde{m}^{00}_{AB}  = \left(\begin{array}{cc} -1 & 0 \\ 0 & -1
\end{array}\right)\\ 
\tilde{m}^{01}_{AB} = \tilde{m}^{02}_{AB} = 0\\
\tilde{m}^{11}_{AB} = \left(\begin{array}{cc} 2 & 1 \\ 1 & 1
\end{array}\right)\\
\tilde{m}^{22}_{AB} = \left(\begin{array}{cc} 1 & 1 \\ 1 & 2
\end{array}\right)\\
\tilde{m}^{12}_{AB} = \tilde{m}^{21}_{AB} = \left(\begin{array}{cc} 1
& 1 \\ 1 & 1 \end{array}\right)
\end{eqnarray}
\endnumparts
and let $m^{ab}_{AB} = \tilde{m}^{ab}_{AB} - \epsilon
\delta^{ab}\delta_{AB}$ for some $\epsilon \ll 1$. Now, by
construction, $m^{00}_{AB}$ is clearly negative definite, so the usual
function $t$ is a time function. A direct computation shows that
$\tilde{m}^{ab}_{AB}\eta_a\eta_b$ is positive definite for any
non-zero $\eta$ that satisfies $\eta(\partial_t) = 0$. Hence for
sufficiently small $\epsilon$, the coefficient matrix $m^{ab}_{AB}$ is
regularly hyperbolic, and the equation 
\[ m^{ab}_{AB}\partial^2_{ab} \psi^B = 0 \]
has a well-posed initial value problem for finite energy initial data.
However, the canonical stress tensor $Q[\psi]^0_0$ is not point-wise
positive definite for all initial data! Take an arbitrary smooth initial 
data with compact spatial support on $t = 0$ such that inside the ball
of radius 1, the data takes the following explicit values:
\numparts
\begin{eqnarray}
\psi^1(0,x,y) = y~, \\
\psi^2(0,x,y) = -x~, \\
\partial_t\psi^A(0,x,y) = 0~.
\end{eqnarray}
\endnumparts
It is easy to check that for this initial data,
$\tilde{m}^{ab}_{AB}\partial_a\psi^A\partial_b\psi^B = 0$ in a
spatial neighborhood of 0. And hence in that neighborhood the energy
density for $Q[\psi]^0_0$ is in fact point-wise negative.

This lack of point-wise positivity is a reflection of the internal
structure of the target manifold $N$. Indeed, if $N$ were
one-dimensional or effectively one-dimensional (i.e.\ the coefficient
matrix $m^{ab}_{AB} = g^{ab}h_{AB}$ is \emph{separable}), two waves
that are spatially coincident and travels in the same direction must
interact, and hence a point-wise measurement can capture correct
notion of energy. But when $N$ has an internal structure, two waves
that are spatially coincident and travels in the same direction does
not have to interact strongly, since they maybe orthogonal \emph{on}
$N$. This orthogonality condition cannot be detected point-wise in the
energy density, but manifests itself as ``miraculous cancellations''
when the total energy is considered. In other words, to detect and
account for this orthogonality requires a spatial mode-decomposition
of the fields, as seen in \cite[Ch.\ 5]{Christ2000}. 

That dominant energy condition cannot guarantee hyperbolicity can be
seen also as a manifestation of this internal structure. Proposition
\ref{prop:equivalence} effectively says that perturbations $\psi$ 
which are \emph{parallel on $N$ to the background solution $\phi$} 
can also be controlled by the dominant energy condition. In the case
where the dimension of $N$ is greater than 1, it is precisely those
additional degrees of freedom, which cannot be accounted for merely 
by the dominant energy condition, that requires the framework of
regular hyperbolicity.

\section{Hyperbolicity of the Skyrme model}\label{sect:ex}
It is easy to see that the Skyrme model can be considered as a theory
of elasticity with the restrictions on the particle world-lines
relaxed. This connection has also been observed by Slobodeanu 
\cite{Slobod-p2010}. From the point of view of relativistic fluids and
elasticity, it then is perhaps less surprising that highly boosted
Skyrmions are expected to be unstable \cite{CruBel1994}. Here we
compute exactly the symbol $m^{ab}_{AB}$ associated to the Skyrme
model and exhibit an ultrahyperbolic type breakdown of hyperbolicity
in tachyonic regimes. 

For the following discussion, we choose the scale factors to normalize
the Skyrme Lagrangian \cite{Skyrme1954, Skyrme1961} to 
\begin{equation}
\fl S = \int \frac12 g^{ab}h_{AB}\partial_a\phi^A\partial_b\phi^B +
\frac14(g^{ab}h_{AB} g^{cd}h_{CD} - g^{ab}h_{AD}g^{cd}h_{CB})
\partial_a\phi^A
\partial_b\phi^B\partial_c\phi^C\partial_d\phi^D~dvol_g
\end{equation}
We ignore the mass term $s(\phi)$ as it plays no role in the
discussion for hyperbolicity; here
$\phi:\mathbb{R}^{1+3}\to\mathbb{S}^3$ is assumed as usual. Applying
Theorem \ref{thm:lagsuffdec} we see that (with or without the positive
mass term) the dominant energy condition is always satisfied for this
model. 

Using \eref{eq:defprincipal}, the coefficients $m^{ab}_{AB}$ can be
computed to effectively be
\begin{equation}
\fl m^{ab}_{AB} = g^{ab}h_{AB}(1 + |\partial \phi|^2_{g,h}) +
\partial^a\phi_A \partial^b\phi_B -
g^{ab}g^{cd}\partial_c\phi_A\partial_d\phi_B -
h_{AB}h_{CD}\partial^c\phi^C\partial^d\phi^D
\end{equation}
From Theorem \ref{thm:existtimefn} we already see that any time
function of Minkowski space is an admissible time function. Therefore
it suffices to examine the conditions for $m^{ab}_{AB}\eta_a\eta_b$ to
be positive definite. 

The following computations are most illustrative in a frame adapted to
$d\phi$. There are two cases: the kernel of $d\phi$ being
a degenerate (null) subspace, and everything else. In the everything else
category, there exists an orthonormal frame of that simultaneously
diagonalizes $g$ and $\phi^*h$; in the case the kernel of $d\phi$ is a
degenerate, there exists an exceptional null frame (see, e.g.\
\cite[Ch.\ 9, exercises 18 and 19]{ONeil1983} or \cite{LanRod2005}, 
with some cases ruled
out by virtue of $\phi^*h$ being positive semidefinite). The null case can be
checked to be hyperbolic much in the same way as described below, so
we shall omit its discussion (note that for the null case the
canonical stress corresponding to $\sigma_2$ is degenerate; already it
requires the presence of the wave-map term $\sigma_1$ in the
Lagrangian for hyperbolicity) and focus on the generic case of an
orthonormal basis. Let the vector basis be $e_0,e_1,e_2,e_3$, orthonormal
with respect to the metric $g$, and diagonal with respect to
$(\phi^*h)$. Let their corresponding covector basis
be $f_0,f_1,f_2,f_3$ where $f_j(e_i) = \delta_{ij}$. And let
$\lambda_0^2,\lambda_1^2,\lambda_2^2,\lambda_3^2$ be the values for
$(\phi^*h)(e_i,e_i)$ respectively. Note at least one of them must 
vanish from rank
considerations. Let $e'_i$ be a unit vector in the tangent space of
the target manifold such that the push-forward $d\phi(e_i) = \lambda_i
e'_i$: for $e_i$ such that $d\phi(e_i)= 0$, choose $e'_i$ such that
one of the $e'_i$ vanishes, and the rest forms an orthonormal basis
for $h$. Take $f'_i$ to be their corresponding covectors. 

Then we can write
\numparts
\begin{eqnarray}
g^{ab} = (- e_0\otimes e_0 + e_1 \otimes e_1 + e_2\otimes e_2 +
e_3\otimes e_3)^{ab} \\
h_{AB} = (f'_0\otimes f'_0 + f'_1\otimes f'_1 + f'_2 \otimes f'_2 + f'_3\otimes f'_3)_{AB}\\
\partial_a\phi^A = (\lambda_0f_0\otimes e'_0 + \lambda_1f_1\otimes
e'_1 + \lambda_2f_2\otimes e'_2 + \lambda_3 f_3\otimes e'_3)_a^A
\end{eqnarray} 
\endnumparts
Let us consider, in this basis, the component $m^{3*}_{AB}$. (The
computations are symmetric in $\{1,2,3\}$, and the component
$m^{00}_{AB}$ is already treated by Theorem \ref{thm:existtimefn}.) A
direct computation shows that 
\numparts
\begin{eqnarray}
\fl\nonumber m^{33}_{AB} = (1 + \lambda_1^2 + \lambda_2^2)f'_0\otimes f'_0
+ (1 - \lambda_0^2 + \lambda_2^2) f'_1 \otimes f'_1 \\
+ (1 - \lambda_0^2 + \lambda_1^2) f'_2\otimes f'_2 +
(1 - \lambda_0^2 + \lambda_1^2 + \lambda_2^2)f'_3\otimes f'_3\\
\fl m^{30}_{AB} = -\lambda_3\lambda_0 (f'_3)_A (f'_0)_B \\
\fl m^{31}_{AB} = \lambda_3\lambda_1 (f'_3)_A (f'_1)_B 
\end{eqnarray}

Now first consider the case where $\lambda_0 = 0$. Then we see that
$m^{33}_{AB}$ is automatically positive definite. This case we have
already discussed in the previous section: if $d\phi$ has a time-like
element in its kernel, then that element is a valid observer field. 

For the case where $\lambda_0 \neq 0$, we consider first when
$\lambda_3 = 0$. Then we see that if $\lambda_0^2 > 1 +
\min(\lambda_1^2,\lambda_2^2)$, we have $m^{33}_{AB}$ no longer
positive. Furthermore, we have by the above computation, that in this
case, for $\eta = s f_3 + r f_0$, $m^{ab}_{AB} = s^2 m^{33}_{AB} + r^2
m^{00}_{AB}$, a sum of a indefinite bilinear form with a negative
definite one. Therefore we see that for any covector $\eta$ in the span of
$f_0$ and $f_3$, we cannot have $m^{ab}_{AB}\eta_a\eta_b$ be positive
definite. In particular, this implies that in this case regular
hyperbolicity must fail. 

For the case where $\lambda_0 \neq 0 \neq \lambda_3$, and $\lambda_1 =
0$, we see that $m^{33}_{AB}$ loses positivity as soon as $\lambda_0^2
> 1$. Now we consider again $\eta = s f_3 + r f_0$. Computing
explicitly we get
\begin{eqnarray*}
\fl m^{ab}_{AB}\eta_a\eta_b = (s^2- r^2)(1 + \lambda_2^2)f'_0\otimes f'_0
+ \left[ s^2 (1 - \lambda_0^2 + \lambda_2^2)- r^2 (1 + \lambda_2^2 +
\lambda_3^2)\right] f'_1 \otimes f'_1 \\
+ \left[ s^2 (1 - \lambda_0^2) - r^2(1 + \lambda_3^2)\right] f'_2\otimes f'_2 
+ (s^2-r^2) (1 + \lambda_2^2)f'_3\otimes f'_3 \\
- (s\lambda_0 f'_3 + r\lambda_3 f'_0)\otimes(s\lambda_0f'_3 + r\lambda_3
  f'_0)
\end{eqnarray*}
where we see that the coefficient of $f'_2\otimes f'_2$ will always
be negative. This implies that for any $\eta$ of this form, when
$\lambda_1 = 0$ and $\lambda_0^2 > 1$, we have that
$m^{ab}_{AB}\eta_a\eta_b$ cannot be positive definite.

From these computations, we get the following theorem
\begin{thm}\label{thm:skyrmehyperregime}
The Skyrme model is regularly hyperbolic when any of the following
is true:
\begin{enumerate}
\item The kernel of $d\phi$ contains a time-like vector.
\item The kernel of $d\phi$ is a degenerate subspace.
\item The time-like eigenvector for $D^\phi$ has a corresponding
eigenvalue with norm less than 1. 
\end{enumerate}
The Skyrme model has an ultrahyperbolic breakdown of hyperbolicity
when $D^\phi$ admits a time-like eigenvector with eigenvalue with norm
greater than 1.
\end{thm} 
\begin{rmk}
The number ``1'' appearing in the third condition depends on the
normalisation chosen for the Lagrangian, which dictates the
interaction strength between the wave map and fluid-like terms. This
result should be compared with the conclusion drawn by Crutchfield and
Bell in \cite{CruBel1994}, where their linear analysis shows that a
sufficient criterion for break down of hyperbolicity is, in the above
language, $\lambda_0^2 > 1 + \lambda_1^2 + \lambda_2^2 + \lambda_3^2$.
So the present result sharpens the regime for which hyperbolicity
fails. Note that this failure of hyperbolicity is not merely that of
the regular hyperbolicity method. In parts of this regime it can be
shown explicitly that the associated linear system is ill-posed; see
the next remark. 
\end{rmk}
\begin{rmk}
Fixing $\lambda_3 = 0$ and $\lambda_1^2 < \lambda_0^2 - 1 <
\lambda_2^2$ (or $\lambda_0^2 > 1 + \lambda_1^2 + \lambda_2^2$, which
is the Crutchfield-Bell condition), it
is easy to check (using Descartes' rule of signs) that the linear, constant
coefficient equation $m^{ab}_{AB} \partial^2_{ab}\psi^B = 0$
cannot be hyperbolic (in the sense that its polynomial symbol has the
requisite number of real roots). Recall that hyperbolicity requires
there to be a hyperbolic direction $\eta$ such that, for any $\zeta$
transverse to $\eta$, the polynomial $M(s) = \det\left[m_{AB}(\zeta +
s\eta, \zeta + s\eta)\right]$ has only real roots \cite{Gardin1959,
Horman1976}. We fix $\zeta = f_3$. Using that $\lambda_3 = 0$, we can
exploit a symmetry condition that if $\eta$ is hyperbolic, so will
$\eta - 2 g(\eta,f_3) f_3$. Using that the
hyperbolic directions form a convex cone \cite{Gardin1959}, we can
assume the would-be hyperbolic direction is orthogonal to $f_3$. For
such an $\eta$, we have that $M(s)$ is an even polynomial. The
computations given before the statement of the theorem implies that
$\lim_{s\to\pm\infty}M(s)/s^6 < 0$, with $M(0) < 0$. Therefore
$M(s)$, a sixth degree polynomial, can have at most 4 real roots. The
same argument, reversing the role of $\eta,\zeta$, can also be used 
to rule out $f_3$ as a hyperbolic direction, proving the claim. 
\end{rmk}

If we assume the local well-posedness result claimed in
\cite{Christ2000}, then the above theorem implies the following
well-posedness property for the Cauchy problem of the Skyrme model.
\begin{cor}
The Cauchy problem for the Skyrme model with \emph{almost stationary}
initial data, where almost stationary is read to satisfy the
hypotheses of Theorem \ref{thm:skyrmehyperregime}, is locally
well-posed. In particular, a small perturbation of a static
Skyrmion configuration gives rise to a well-defined evolution. 
\end{cor}
\begin{rmk}
The existence of static configurations to the Skyrme model is a
partially open problem. In certain symmetry classes the existence is
known, see \cite{Esteba1986, KapLad1983, LinYan2004}. 
In the case of the perturbation of a static configuration, local
well-posedness also follows, by the computations above, using the
techniques of \cite{HuKaMa1976}. The framework of regular 
hyperbolicity is not necessary in that regime. 
\end{rmk}
\begin{rmk}
It has been pointed out to the author by Dan Geba that one
automatically has local well-posedness with arbitrary smooth initial
data, if one were to consider the \emph{spherically symmetric} Skyrme
model. Indeed, once the symmetry is imposed, the target space (being a
quotient of a three dimensional manifold by a symmetry with two
dimensional orbits) is effectively one dimensional, and the
ultrahyperbolic-type breakdown which is due to the \emph{internal
structure} of the target manifold cannot occur (see also Section
\ref{sect:canstrvseh}). 

However, in this situation one has a hyperbolic illustration of the
failure of the \emph{Coleman principle}, analogous to that observed by
Kapitanski and Ladyzhenskaya \cite{KapLad1983}. More precisely, the
existence of uniqueness of solutions in the spherically symmetric
class does not automatically guarantee that said solutions are in fact
unique in the non-spherically symmetric class. This is because the
spherically symmetric solution is only guaranteed to be stable under
spherically symmetric perturbations. In our case, however, the above
analysis shows that the linearized system is not necessarily stable
under asymmetric perturbations, and therefore we cannot use Cauchy
stability to conclude that symmetric initial data must lead to
symmetric solutions!
\end{rmk}


\ack The author is supported by the Commission of the European
Communities, ERC Grant Agreement No 208007. The author would like to thank 
Nick Manton and Claude Warnick for introducing him to the problem, and 
to Mihalis Dafermos, Gary Gibbons and Jared Speck for useful discussions.

\bibliographystyle{unsrt}

\end{document}